\newcommand*{\Fversion}{}%

\documentclass[11pt,a4paper]{article}
\usepackage[margin=1in]{geometry}
\usepackage[T1]{fontenc}
\usepackage{xcolor}
\usepackage{graphicx}
\usepackage{amsthm}
\usepackage{authblk}
\usepackage{refcount}
\definecolor{linkcolor}{RGB}{170,0,0}
\definecolor{citecolor}{RGB}{44,160,46}
\definecolor{urlcolor}{RGB}{0,51,153}

\usepackage{doc}
\usepackage{enumitem}

\usepackage[textwidth=3cm,tickmarkheight=3pt,disable]{todonotes}

\definecolor{citeblue}{rgb}{0.1,0,.4}

\RequirePackage{xspace}
\usepackage{amsmath}
\usepackage{amssymb}
\usepackage{amstext}  
\usepackage[ruled,vlined]{algorithm2e}

\usepackage{bussproofs}

\newtheorem{observation}{Observation}

\DeclareMathOperator*{\opos}{\mathsf{OutPos}}
\DeclareMathOperator*{\npos}{\mathsf{PrePos}}

\newcommand{\arity}{\mathfrak{a}}
\newcommand{\solve}{\texttt{\textup{SolveByAlternation}}}
\DeclareMathOperator*{\mgu}{\mathrm{mgu}}

\newcommand{\class}[1]{\ensuremath{\mathsf{#1}}}
\newcommand{\Pe}{\class{P}}
\newcommand{\NP}{\class{NP}}
\newcommand{\PH}{\class{PH}}
\newcommand{\NPSPACE}{\class{NPSPACE}}
\newcommand{\PSPACE}{\class{PSPACE}}
\newcommand{\Sigmap}[1]{\class{\Sigma_{#1}^{\class p}}}
\newcommand{\Pip}[1]{\class{\Pi_{#1}^{\class p}}}
\newcommand{\EXPTIME}{\class{EXPTIME}}
\newcommand{\NEXPTIME}{\class{NEXPTIME}}

\newcommand{\NL}{\class{NL}}

\newcommand{\horn}{\textsf{Horn}\xspace}
\newcommand{\krom}{\textsf{Krom}\xspace}
\newcommand{\affine}{\textsf{Affine}\xspace}
\newcommand{\dethorn}{\textsf{Det}-\horn}
\newcommand{\bs}{Bernays-Sch\"onfinkel\xspace}
\newcommand{\fragment}{QEALM-fragment\xspace}

\DeclareMathOperator*{\cnf}{\phi}
\DeclareMathOperator*{\qbf}{\Phi}
\usepackage{tikz}
\usepackage[colorlinks=true,urlcolor=urlcolor,citecolor=citecolor,linkcolor=linkcolor]{hyperref}

\newtheorem{theorem}{Theorem}
\newtheorem{lemma}{Lemma}

\newtheorem{example}{Example}
\newtheorem{definition}{Definition}

\newcommand{\ecalculus}{$\forall$\textsf{Exp+Res}\xspace}

\title{Reintroducing the Second Player in EPR}

\author[1]{Leroy Chew}
\author[1]{Mikol\'a\v{s} Janota}
\author[2]{Miroslav Ol\v{s}\'{a}k}
\author[1]{Martin Suda}
\affil[1]{Czech Technical University in Prague, Prague, Czech Republic\protect\\
  \texttt{\{leroyche,mikolas.janota,martin.suda\}@cvut.cz}}
\affil[2]{Charles University, Prague, Czech Republic\protect\\
  \texttt{mirek@olsak.net}}
\date{}

\begin{document}

\maketitle

\begin{abstract}
	We present a PSPACE-complete fragment of first order logic  we call QEALM (QBF-like EPR by Alternating Level-ordered Miniscoping). Unlike other well-known fragments it is closer to QBF (Quantified Boolean Formulas) in nature because it has game-like alternations and can be restricted to complete problems for the polynomial hierarchy. QEALM is a natural first-order analogue to QBF just as EPR (Effectively Propositional logic) is a first-order analogue to Dependency QBF. 
	
	We prove the PSPACE and Polynomial Hierarchy completeness theorems of QEALM along with several other nice properties such as closure under Robinson's resolution and retained hardness after intersection with Schaefer fragments. 
	We identify problems in the TPTP library that fall into this fragment and their level in the polynomial hierarchy.
\end{abstract}

\section{Introduction}

In automated reasoning, translations to formal logic are an effective tool for efficient solving, trusted proving and a solid theoretical backbone. 
Whether to choose a propositional or first-order logic encoding can depend on the problem.
In propositional logic all objects are 
binary and is often sufficient for computational tasks.
But what is possible may not always be practical, and First-Order Logic (FOL) offers quantifiers, variables and predicates. 
Satisfiability for propositional logic is \NP-complete~\cite{CookLevin71}, whereas unsatisfiability for FOL is only semi-decidable.
And the computational complexity of the problem can often dictate the best choice for encoding.

If we fix our focus to \NEXPTIME,  we can find complete problems using either propositional or first-order ideas. 
The \bs fragment is nearly a century old~\cite{ram30} and defines the FOL fragment that is free of function symbols and existential quantification.
The \bs fragment is also known as Effectively Propositional (EPR), because of an exponential reduction to propositional logic, 
hence its satisfiability is \NEXPTIME\ for which it is also complete.
For propositional logic, we obtain \NEXPTIME-completeness by first extending it with Boolean quantifiers. This gives  the \PSPACE-complete Quantified Boolean Formulas (QBF). To get \NEXPTIME-completeness, one allows Henkin quantification \cite{Henkin61}, which relaxes the quantifier ordering to non-transitivity which defines the Dependency QBF (DQBF) problem. 
This means, for \NEXPTIME\ problems, we have a choice of a Boolean or a first-order encoding.

In this paper 
we aim to have the same dichotomy for \PSPACE. We can already extend propositional logic to the  \PSPACE-complete QBF, so  
our aim is to define an analogous \PSPACE-complete fragment of FOL .
Our proposed fragment (eventually called QEALM) will share some similarities to QBF, but has the advantage of being able to handle non-Boolean elements.
Because QEALM is a sub-fragment of \bs, (paralleling how QBF extends to DQBF) it means that the game evaluation is not done in the quantifier prefix, as all variables are universal.
Instead, the ordered list of arguments for each predicate symbol, takes the place of a prefix. To achieve this in \PSPACE\ we have restricted QEALM to CNF problems, where each clause must have the QEALM condition (Definition~\ref{def:fragment}). 
Therefore the union and intersection of QEALM problems is also a QEALM problem. This is much like how the Schaefer fragments
\horn, \krom, \affine  \cite{Schaefer1978}  operate.
Unlike Schaefer classes but similar to QBF, the fragment can be easily restricted to fixed alternations to get complete problems for the polynomial hierarchy.

\subsection{Related Work}

We note several connections between our work and other work in first-order logic.
Our idea is similar to the \emph{fluted fragment} of first-order logic \cite{quine1969limits}, which gains decidability 
by requiring the order of quantification to match the order of variables as arguments in predicates.
As we stay in \bs, our work removes the quantifier order entirely from the classification.

Our work is directly inspired by translations from QBF to FOL, in particular translations to \bs directly (\cite{SLB12}).
We describe this translation in Section~\ref{sec:qbf2epr}.
In addition, our detection tool uncovered that our new fragment is not only an extension of a FOL translation of QBF, but 
can be understood 
as extensions of FOL translations from other \PSPACE\ logics, in particular 
\emph{basic path logic} \cite{DBLP:journals/jar/Schmidt99}, a fragment of FOL introduced as a target of the optimised functional translation \cite{DBLP:journals/logcom/OhlbachS97} of certain modal logics. 
While the motivations behind that work are different, this is an interesting connection, that we intend to explore in future work.

It should be noted that 
\PSPACE\ has been studied previously in a FOL setting. There are a number of known \PSPACE-complete Schaefer fragments of EPR as detailed in Table~\ref{fig:complexity}.
There are two common kinds of \PSPACE-complete problems: problems from unbounded alternation and problems from reachability.
QBF is naturally of the former kind, whereas some well-known \PSPACE\ problems such as planning are of the latter.
\krom and \dethorn more naturally capture reachability and planning as they allow us to encode transition rules
and start and target states as clauses, where unsatisfiability equals reachability. %
Morally, we can informally think of EPR-\krom and EPR-\dethorn in \textsf{CoNPSPACE} more so than \PSPACE,
as these were \textsf{CoNL}-complete propositional problems (\textsf{NL}=\textsf{CoNL}) before we took them to first order and used Savitch's Theorem~\cite{Sav70} for the quadratic transformation to QBF. 
Thus we can justify the advantages of our fragment over \krom and \dethorn.
A more concrete example is that it is not clear how to restrict  EPR-\krom and EPR-\dethorn
to complete classes of the polynomial hierarchy \PH~\cite{Stockmeyer76}, as QBF naturally does.

Much like \horn can be generalised into HORN-renamable, based on the symmetry of true and false. Problems that are not QEALM can become so under the application of a symmetry. 
\subsection{Organisation}
This paper first presents the QEALM fragment (Definition~\ref{def:fragment}) and proves its \PSPACE-completeness (Theorem~\ref{thm:completeness}). Naturally, we show both inclusion (Lemma~\ref{lem:memory}) and hardness in \PSPACE\ (Lemma~\ref{lem:hardness}). 
We show that the intersection of QEALM with Schaefer fragments retains \PSPACE-hardness (Theorem~\ref{thm:intersection}). 
In contrast to thes Schaefer fragments, our fragment QEALM will have a game-like alternation, where we can restrict the number of alternations to specific levels of the polynomial hierarchy. Thus we show completeness for \PH\ as well (Theorem~\ref{thm:phcompleteness}). We additionally show its closure under Robinson's resolution (Theorem~\ref{thm:closure}), a property that separates it from QBF and previous translations from QBF into FOL.%

%

\begin{table}[h]
	\centering
	\begin{tabular}{|l|r|r|}
		\hline
		Fragment & Prop. complexity& EPR complexity\\
		\hline
		3CNF & \NP-complete \cite{CookLevin71}& \NEXPTIME-complete \cite{lewis1980}\\
		\horn & \Pe-complete \cite{Jones1976} & \EXPTIME-complete \cite{Plaisted1984}\\
		\dethorn & \NL-complete \cite{JLL1976} &\PSPACE-complete \cite{Plaisted1984}\\
		\krom (2-CNF)& \NL-complete \cite{JLL1976}& \PSPACE-complete \cite{Plaisted1984}\\
		\affine & \textsf{Mod$_2$-L}-complete \cite{BDHM92} & \textsf{Mod$_2$-}\PSPACE  \\
		\hline
	\end{tabular}
	\caption{Complexity of EPR fragments\label{fig:complexity}.} %
\end{table}

	\section{Preliminaries}

We assume familiarity with classical first-order logic~\cite{smullyan}.
In particular, terms are built from
variables and function symbols; each function and predicate
symbol $f$ has a fixed arity $\arity(f)$. Function symbols with arity 0 are
called constants. A \emph{literal} $l$ is an \emph{atom}, e.g.,\ $l=r(t_1 \dots
t_{\arity(r)})$, or a negation of an atom $l=\lnot r(t_1 \dots t_{\arity(r)})$. A
\emph{clause} is a disjunction of literals, e.g.,\ $C= l_1 \vee l_2 \vee l_3$, and
a \emph{CNF} (conjunctive normal form) is a conjunction of clauses, e.g.,\ $C_1
\wedge C_2 \wedge C_3$, where variables occurring in the clauses are implicitly
assumed to be universally quantified.
The quantifiers $\forall$ and $\exists$ bind a variable
$x$ and act on a formula $\phi$, i.e.,\ $\forall x \phi$. A sentence is a formula
where all variables are bound by some quantifier. A sentence is in
\emph{prenex} form if no subformula is defined by applying a connective
($\lnot$, $\wedge$, $\vee$) to a quantified subformula.

The Bernays-Schoenfinkel fragment comprises sentences of the form
$\exists^*\forall^*\phi$, where $\phi$ is quantifier-free and terms are only
variables. Equivalently, via Skolemization, it is of the form $\forall^*\phi$,
where terms are only constants and variables, i.e.,\ there are no function
applications but predicates are allowed,  which is the view we adopt in this
paper. Moreover, we can also assume $\phi$ is in CNF, as the standard
clausification procedures preserve the fragment.

An interpretation is a tuple $(D, I)$ where $D$ is a non-empty set, called the
domain, and $I$ prescribes values for function and predicate symbols.
Whenever convenient, we view an interpretation for a predicate as a function to
$\{0, 1\}$. In \bs, there are no function symbols of arity greater than zero, and therefore the Herbrand
universe is finite---it is the set of constant symbols that appear in the
formula; this shows the decidability of \bs via the Herbrand theorem. Throughout
the paper, we therefore assume that $D$ is in fact the set of constants of the
formula in question.

As usual, a \emph{substitution} $\sigma$ is a finite mapping from variables to terms
which, when applied to a formula $\phi$, gives rise to the formula $\phi\sigma$
in which each free occurrence of a variable $x$ from the domain of $\sigma$ gets (simultaneously) replaced
by the corresponding term $t=\sigma(x)$ from $\sigma$'s range.
We refer to the resulting formula $\phi\sigma$ as an \emph{instance} of $\phi$,
and to the operation (and where confusion cannot arise, also to the substitution itself) as \emph{instantiation}.

\paragraph{Robinson's Resolution}
Given a CNF formula $\phi$, Robinson's Resolution construed as a proof system consist of the Axiom rule
\begin{prooftree}
	\AxiomC{}
	\RightLabel{(Axiom)}
	\UnaryInfC{$C$}
\end{prooftree}
where $C$ is a clause from $\phi$, and two derivation rules, Resolution and Factoring, defined as follows:
\begin{prooftree}
	\AxiomC{$C_1 \vee \neg p(t_1 \dots t_n)$}
	\AxiomC{$C_2 \vee p(s_1 \dots s_n)$}
	\RightLabel{(Resolution)}
	\BinaryInfC{$C_1 \sigma\vee C_2 \sigma$}
\end{prooftree}
where the atoms $p(t_1 \dots t_n)$ and $p(s_1 \dots s_n)$ are unifiable and  $\sigma=\mgu(p(t_1 \dots t_n),p(s_1 \dots s_n) )$ is their most general unifier. (The premises $C_1$ and $C_2$ are assumed to not share variables. This can be achieved by first renaming the premises apart, if necessary.)
\begin{prooftree}
	\AxiomC{$C \vee p(t_1 \dots t_n) \vee p(s_1 \dots s_n)$}
	\RightLabel{(Factoring)}
	\UnaryInfC{$C\sigma\vee p(t_1 \dots t_n)\sigma$}
\end{prooftree}
where $\sigma=\mgu(p(t_1 \dots t_n),p(s_1 \dots s_n) )$.

Propositional resolution is a simplification of Robinson's Resolution. Assuming clauses are sets of literals, we do not need a factoring rule, and the resolution rule does not need to consider unification, the pivot atoms must match exactly. We can use propositional resolution on FOL problems that are fully ground as unification is no longer possible.  
\subsection{Observations on First-Order Logic} 

The satisfiability problem for EPR requires us to find at least one suitable interpretation of the predicate symbols. Therefore, unsatisfiability requires that there are no suitable interpretations, in other words, for every possible interpretation there is some assignment to the universal variables that makes the formula false.
In the general case, the falsifying universal assignment depends on and \textit{after} the chosen interpretation of the predicate symbols and---when considering uninterpreted predicate symbols---several different assignments might be needed to demonstrate unsatisfiability.
However, there are certain scenarios where a common universal assignment falsifies all interpretations of the predicate symbols. Consider the following:
\begin{example}\label{ex:terms}
	Consider a \bs formula, with constants $c_1$ and $c_2$:%
	$$
	\begin{aligned}
		\forall x, y &\quad & \big(p(x)\vee \lnot q(x,c_1)\big) \quad &\wedge& \big(\lnot p(x) \vee \lnot q(x,c_2)\big)\\
		\big(\lnot r(x)\big) \quad& \wedge& \big( q(c_1,y)\vee r(c_1) \big)	\quad&\wedge& \big(q(c_2,c_1)\big)
	\end{aligned}
	$$
	and a propositional resolution refutation over ground instances:
	
	{\centering\begin{tikzpicture}
		\node(l1) at (-4,4){$q(c_1,c_2)\vee r(c_1)$};
		\node(c1) at (0,4){$\lnot r(c_1)$};
		\node(r1) at (4,4){$q(c_1,c_1)\vee r(c_1)$};
		\node(l2) at (-1,3){$q(c_1,c_2)$};
		\node(r2) at (1,3){$q(c_1,c_1)$};
		\node(l3) at (-4,3){$\lnot p(c_1)\vee \lnot q(c_1,c_2) $};
		\node(r3) at (4,3){$p(c_1)\vee \lnot q(c_1,c_1) $};
		\node(l4) at (-2,2){$\lnot p(c_1)$};
		\node(r4) at (2,2){$p(c_1)$};
		\node(c5) at (0,1){$\bot$};
		
		\draw(l1)--(l2)--(c1);
		\draw(r1)--(r2)--(c1);
		\draw(l3)--(l4)--(l2);
		\draw(r3)--(r4)--(r2);
		\draw(l4)--(c5)--(r4);
		
	\end{tikzpicture}\par}
	
	Our main observation is that the first argument in every literal is always the same ground term, in this case $c_1$. This is true, due in part to resolution steps requiring matching arguments on the pivots. Generally the pivot arguments do not match with the other literals in the clause,   but in this formula we have insisted on it in our axioms and it persists throughout the proof. Despite $q(c_2,c_1)$ being an axiom, it ends up not being used in the proof because it is impossible to resolve it with any clause with a mismatch in the first argument. 

	In contrast, the second argument of $q$ not only can contain $c_1$ and $c_2$ in the proof, but is actually required to use both for the refutation to go through. This is because it can be joined by $p(x)$ which does not necessarily share the same argument with the second argument of $q$ in $p(x)\vee \lnot q(x,c_1) $ and $\lnot p(x) \vee \lnot q(x,c_2)$. In other words, because there is some connecting clause that mixes elements, both can be connected in the proof.
\end{example}

\begin{observation}\label{lem:semantic}
	
	Suppose we have a first-order problem (without equality) written as a CNF where we have the following restrictions:
	\begin{itemize}
		\item  only allow predicate symbols which have an arity of at least $1$,
		\item \emph{within each clause $C$} among all literals throughout the clause, the first argument must be the same term $t_C$  (whether it be a variable or other term).
	\end{itemize}

	If the above conditions are met and the formula is unsatisfiable then there is at least one global ground term $t$,
	such that the following transformation produces another unsatisfiable problem
	\begin{itemize}
		\item If possible, $C\in \phi$ is instantiated by the minimal substitution such that the first argument $t_C$ in every literal becomes $t$.
		\item If not possible ($t_C$ cannot unify with $t$), we remove $C$ entirely from the CNF. 
	\end{itemize}

\end{observation}

In other words, having an argument that matches everywhere within clauses, the CNF can be simplified (with the choice of the right ground term) to an equisatisfiable case
in which the argument matches everywhere within \emph{and between} clauses via instantiation and pruning. Here we mention the \textit{first} argument of each literal, but this can be generalised symmetrically under a rearrangement of the arguments of the predicate symbols.

While this observation is semantic in nature, our best understanding of its correctness comes from using the completeness and soundness of Resolution, as every unsatisfiable CNF has a ground refutation, and we can use that to prove the existence of $t$.
This observation does not require us to be in \bs, but we must forbid equality, because the conditions will automatically fail for symmetry and transitivity.

\begin{proof}
	We can use the fact that for unsatisfiable first-order clausal problem, grounding a formula and using propositional resolution on the ground clauses is a sound and complete proof process.
	We further use the fact about trimming Resolution proofs, that proofs follow a DAG-like structure and any clauses not connected to a path to the empty clause can be trimmed.
	
	We will prove by induction on the number of lines in a ground resolution derivation that the first argument of every literal is the ground term $t$.
	
	\noindent\textbf{Base Case: } If we have a derivation of length 1, it just contains a single grounded axiom clause. Because of our  condition, the first argument for every literal is the same ground term.
	
	\noindent\textbf{Inductive Step: }	
	If we have a derivation of clause $C$ after $n$ lines, we have two subderivations of length $<n$ that derive the two clauses that we resolve. So we can assume by induction hypothesis each of the subderivations have a common first argument ground term $t$ within each subderivation. But this must also be common between the two subderivations otherwise the two pivot literals (that contain at least one argument) cannot resolve.

	 Therefore, by induction, all axiom ancestors of the empty clause must have the same term in the first argument for every clause. 
	 By the soundness of propositional Resolution this means there is some ground term that can replace the variable in first argument everywhere via instantiation and for us to still have a contradiction.
	We can then remove all clauses that are not required in the proof and it will still be a contradiction.
\end{proof}

\subsection{Quantified Boolean Formulas (QBF) and Dependency Quantified Boolean Formulas (DQBF)}
The Boolean quantifiers $\forall_B$ and $\exists_B$ can be used to extend propositional logic. They are not to be used in the same way as quantifiers in first-order logic. Instead we quantify a propositional atom symbol that appears without Boolean quantification in a formula (free) over the domain $\{\top, \bot\}$, where  $\top, \bot$ are true and false in this domain (we use the subscript to distinguish it from the symbols in FOL).
These formulas are called QBFs (quantified Boolean formulas) and have the following meaning:
$$ \forall_B x \phi := \phi[\bot/x] \wedge \phi[\top/x]\quad  \exists_B x \phi:= \phi[\bot/x] \vee \phi[\top/x]$$

Like in other languages, algorithms exist to prenex a QBF. Additionally QBF satisfiability can be understood as the truth of the QBF sentence where all free propositional atoms are placed into an outer existential quantifier.
This problem is \PSPACE-complete. In fact, it is understood as canonically \PSPACE-complete. \PSPACE\xspace is a subset of \NEXPTIME, therefore 
more 
solving techniques and reductions are available for QBF satisfiability than with \bs.

One observation about QBFs is that they can be thought of as semantic games between a universal and an existential player.
In such a game, the turns are given by the prefix from left to right.
Whoever's quantifier appears next gets to set a Boolean value ($\{\bot ,\top\}$) to replace the quantified atom.
At the end of the game, the universal player wins if the resulting expression evaluates to $\bot$ and the existential player wins if the resulting expression evaluates to $\top$.
A QBF is true if and only if the existential player has a winning strategy and false if and only if the universal player has winning strategy.
These strategies can be thought of as the Skolem and Herbrand functions, respectively, which are now Boolean functions.
A QBF is true (false) if and only if there are Skolem (Herbrand) functions for the existentially (universally) quantified atoms
and the formula becomes a tautology (contradiction) after substituting in the Skolem (Herbrand) functions.
The quantifier order matters as it affects the size of the domain of the Skolem (Herbrand) functions.

\subsubsection{DQBF and the connection to EPR}\label{sec:qbf2epr}

One extension of prenex QBF that is relevant to this paper is (S-form) DQBFs (Dependency QBFs). In DQBF we once again have a propositional problem with a quantifier prefix, however there is no intrinsic order of the prefix. Instead, S-form DQBFs use Henkin quantifiers for $\exists_B$ that explicitly specify the universal atom symbols that the quantified existential atom depends on. 
The \emph{dependency set} $D_e$ denotes the set of universal variables that existential variable $e$ depends on.
A DQBF is still true if and only if it is a tautology after Skolemisation of the existential variables and every QBF can be written as a DQBF. But the linear two player game can no longer apply in general, so it is difficult to make a simple branching procedure. DQBF turns out to be \NEXPTIME-complete and equivalent to the \bs fragment. %
Let us now describe a transformation \cite{SLB12} from (D)QBF to EPR that uses Skolemisation.
Translating DQBF truth ($\top$) and falsity ($\bot$) as FOL constants \textsf{true} and \textsf{false}, we can express the Skolemisation using a first-order predicate symbol $p$ and function symbols for Skolem functions:
\begin{align*}
	\textlbrackdbl \top \textrbrackdbl_p = p(\textsf{true}), \quad & \textlbrackdbl e  \textrbrackdbl_p= p(f_e(D_e)),\quad& \textlbrackdbl u  \textrbrackdbl_p= p(u),\quad & \textlbrackdbl \phi_1 \lor \phi_2   \textrbrackdbl_p= \textlbrackdbl \phi_1 \textrbrackdbl_p \vee \textlbrackdbl \phi_2 \textrbrackdbl_p, \\
	\textlbrackdbl \bot \textrbrackdbl_p = p(\textsf{false}),\quad&  \textlbrackdbl \lnot e  \textrbrackdbl_p= \lnot p(f_e(D_e)),\quad& \textlbrackdbl \lnot u  \textrbrackdbl_p= \lnot p(u),\quad & \textlbrackdbl \phi_1 \land \phi_2 \textrbrackdbl_p= \textlbrackdbl \phi_1 \textrbrackdbl_p \land \textlbrackdbl \phi_2 \textrbrackdbl_p,
\end{align*}
where $u$ is a universal variable and $e$ an existential variable with its dependency set $D_e$. For each existential $e$ there is a newly introduced symbols $f_e$ of arity $|D_e|$
and its arguments in $f_e(D_e)$ are the corresponding universal variables from $D_e$ in some fixed order.

 Suppose we have a DQBF with universal variables $U$, existential variables $E$ and propositional matrix $\phi$, then we can express the problem of whether the Skolemisation results in a tautology in the following way in FOL: 
 $ \forall U \textlbrackdbl \phi \textrbrackdbl_p \land p(\textsf{true})\land \lnot p(\textsf{false})$.
 
This is not yet in the \bs fragment as we have function symbols for the Skolem functions. These can easily be transformed into predicate symbols as they always occur immediately inside predicate $p$.
Thus $p\circ f_x$ can be replaced by its own predicate symbol $x(\ldots)$ of arity $|D_x|$. Once we complete this for all Skolem functions, the problem lands in EPR.\footnote{
To see this transformation preserves satisfiability, we interpret $x$ via $p\circ f_x$ and, in the less obvious direction, use the fact that only a two-valued domain can be used (cf. Lemma 3 in \cite{SLB12}).}

Finally, we can refine the formulas to get rid of the predicate symbol $p$ by resolving away all literals involving $p$ against the units $p(\textsf{true})$, $\lnot p(\textsf{false})$.\footnote{
Again it is easy to see this is a sound step (via soundness of resolution), but the fact it does not introduce spurious counterexamples is less obvious.
An indirect argument goes via completeness of the QBF calculus \ecalculus \cite{JM15,BCJ19}.}
This leaves all predicate symbols to be the Skolem functions and the arguments among literals will belong to a certain structure.
We generalise this structure for our PSPACE fragment QEALM in the next section.

\section{A QBF-like fragment of \bs}\label{sec:frag}

If the \bs fragment corresponds to the DQBF problem, what fragment of first-order logic corresponds to QBF? Some obvious conditions would be that it would itself be a fragment of \bs and then since QBF is \PSPACE-complete this fragment we would want this also to be \PSPACE-complete.

We define our fragment utilising the condition of Observation~\ref{lem:semantic}, which allows us to have a QBF-like semantic game. For this, we need to fix an ordering of each predicate's arguments. We will see later that there are some practical opportunities to relax this ordering (Definition~\ref{def:rearrange}).

\begin{definition}[QBF-like EPR via Alternating Level-ordered Miniscoping (\fragment)]\label{def:fragment}
	A \bs instance is in the \fragment if the following holds:
	
	In each clause if a variable symbol $u$ appears in multiple literals $l_1 \dots l_k$ (this is not necessarily all literals in the clause),
	\begin{enumerate}
		\item the first position it occurs (i.e.\ the lowest index $i$) is the same in all $l_1 \dots l_k$,
		\item for all $0\leq j\leq i$ for any $1\leq p \leq q\leq k$ the $j$-th argument of  $l_p$ must be equal to the $j$-th argument of $l_q$, i.e., there is a shared initial segment.
	\end{enumerate}
	
\end{definition}

\begin{example}\label{ex:game1pre}

	We take the \bs problem with constants $c_1,c_2$:
	$$\begin{aligned}
		\forall u_1, u_2, u_3, u'_3, u_4 \quad(y(u_1,u_2,u_3,u_4))\wedge\\ 
		(\lnot x(u_1,c_1,u_3,c_1)\vee \lnot y(u_1,c_1,u_3,u_4) \vee \lnot x(u_1,c_1,u'_3,u_1)) \wedge \\ ( x(u_1,u_1,u_3,c_1)\vee \lnot y(u_1,u_1,u_3,u_4) \vee x(u_1,u_1,c_2,u_1))
	\end{aligned}$$
	
	Each clause has a prefix tree. For example in the second and third clauses are described by trees in Figure~\ref{fig:ptree}.
	\begin{figure}[h]
	\begin{tikzpicture}
		\node(uc) at (-2,0){$u_1 c_1$};
		\node(uc3) at (0,0.5){$u_1 c_1 u_3$};
		\node(ucuu) at (0,-0.5){$\neg x(u_1 c_1 u'_3 u_1)$};
		\node(uc34) at (2,1){$\neg y(u_1 c_1 u_3 u_4)$};
		\node(uc33) at (2,0){$\neg x(u_1 c_1 u_3 c_1)$};
		\draw(ucuu)--(uc)--(uc3);
		\draw(uc33)--(uc3)--(uc34);
	\end{tikzpicture}\qquad\qquad
	\begin{tikzpicture}
		\node(uc) at (-2,0){$u_1 u_1$};
		\node(uc3) at (0,0.5){$u_1 u_1 u_3$};
		\node(ucuu) at (0,-0.5){$ x(u_1 u_1 c_2 u_1)$};
		\node(uc34) at (2,1){$ x(u_1 u_1 u_3 c_1)$};
		\node(uc33) at (2,0){$\neg y(u_1 u_1 u_3 u_4)$};
		\draw(ucuu)--(uc)--(uc3);
		\draw(uc33)--(uc3)--(uc34);
	\end{tikzpicture}
	\caption{\label{fig:ptree} The prefix trees governing the second(left) and third(right) clauses in Example~\ref{ex:game1pre}}
	\end{figure}
	
	The arguments of the literals appear as nodes on tree. Here all literal arguments appear on the leaves. 
	Variables that are first introduced on one branch such as $u_4$ cannot be introduced in another branch as the same variable, however duplicate variables from earlier in the tree can be reused, and constant symbols can appear on any branch. The leaves are the literals,  we adopt the convention of a  distinction between different unique literals, so when the arguments are the same but we use different predicate symbols we make a final branching. 
	
	The idea is this tree is traversed over the course of a two player game similar to the game found in QBF. 
	In QBF the game follows the linearly ordered quantifiers (for prenex formulas). 
	The tree branching dictates the existential plays choices, while the value of the universal variables along those branches is decided by the universal player.
	Because we intend to solve such formulas by instantiating the variables in order, after the first occurrence repeated variables act more like constants and this is reflected in the definition of membership in the \fragment.

\end{example}

The fragment is a generalisation of how QBFs (written as DQBFs) translate under the transformation in Section~\ref{sec:qbf2epr}. However one key difference, is that the literals can exhibit variable disjointness within a clause. This is not really necessary for \PSPACE-hardness, as translated QBFs are already \PSPACE\ hard, but it will be necessary later in Theorem~\ref{thm:closure} for closure under Robinson's Resolution.

Let us now show that the \fragment has an algorithm that decides its satisfiability in \PSPACE. The idea of the algorithm is to recursively alternate between a universal loop that checks instantiations for falsity, and an existential loop that queries branches on disjunctions while checking for satisfiability. The branching for variable disjointness seen in Example~\ref{ex:game1pre} and Figure~\ref{fig:ptree} is an important feature, because miniscoping can split the quantifiers up onto separate subclauses and then an existential loop can check whether any miniscoped subclauses are satisfiable. Likewise, a non-empty set of arguments shared among the literals of a clause will allow the universal loop to query a joint instantiation to the shared variables. Before we describe how this works,
let us first define the variables that the universal loop of the algorithm will attempt to assign.

\begin{definition}[Outer Variables]\label{def:outer}
	\emph{Outer variables} of a clause are those that appear in every 
	literal of the clause.
\end{definition}

These are named outer variables, because---as we will see---while all other variables can be handled by nested miniscoping, the outer variables necessarily scope the entire clause.
We eventually want variables that scope the entire CNF, but because conjunction commutes with the universal quantifier, variables in different clauses may as well be completely different. Instead we look for positions of arguments that are amenable to Observation~\ref{lem:semantic}.

\begin{definition}[Outer Position]\label{def:outer2}
	\hspace{2cm}
	\begin{itemize}
		\item For a clause $C$, $i$ is an \emph{outer position} ($i\in \opos(C)$) if and only if there is an outer variable $u$ such that $i$ is the least position of an occurrence of $u$ in any literal in $C$.
		
		\item For a clause $C$, $i$ is a \emph{prefix position} ($i\in \npos(C)$) if and only if $C$ is a singleton, or for each $j\leq i$ there is a term $t_j$ that is either a constant or an outer variable, such that the $j$-th position of every literal in $C$ contains $t_j$.
		
		\item We overload the notation $\opos(\phi)$ for a CNF $\phi$. Position $i$ is an \emph{outer position for CNF}  $\phi$ ($i\in \opos(\phi)$) if there is some clause $C$ where $i$ is an outer position, and in all other clauses $D$, position $i$ is a prefix position \\($\exists C \in \phi$ s.t. $ i\in \opos(C)$ and $i \in \bigcap_{D\in \phi}\npos(D)$).
	\end{itemize}
	
\end{definition}

Global outer positions are those that appear outer in at least one clause, but they must  be in the joint argument prefix of every other clause otherwise they might be assigned prematurely.  

\begin{example}
	Let $u_1 \dots $ denote universal variables, $c_1 \dots $denote constant symbols and $p_1 \dots $ denote predicate symbols.

		$$\begin{aligned}
		\forall u_1, u_2, u_3, u_4 \quad\big(p_1(u_1,c_1,u_1,u_2)\lor p_2(u_1,c_1,u_1,u_2,c_2)\lor p_3(u_1,c_1,u_1,u_2,u_3, u_4)\big)\wedge\\ 
		\forall u_1\quad\big(\lnot p_1(c_1,u_1,c_3,u_4)\vee  \lnot p_2(c_1,u_1, c_4, c_3, c_1)\big)
	\end{aligned}$$
	
	In the first clause, positions 1,2, 3 and 4 are prefix and positions 1 and 4 are outer.
	In the second clause position 1 is 
	prefix and position 2 is outer.
	Therefore positions 1 and 2 are outer globally.

\end{example}

\begin{lemma}\label{lem:ulem}
	Suppose we have a CNF $\phi$ in the \fragment and for every $i\in \opos(\phi)$, we have a constant symbol $c_i$. 
	Let us build a substitution $\sigma$ such that for every variable $u$ that appears in some argument  $i\in \opos(\phi)$ for some literal occurring in $\phi$, we replace $u$ with $c_i$.
	When we apply that substitution to $\phi$, the resulting $\bigwedge_{C\in\phi } C\sigma$ remains in \fragment.
	
\end{lemma}

\begin{proof}
	When we instantiate a variable with a constant, we replace it everywhere. Therefore any initial argument that agreed across literals before the instantiation remains in agreement.
	Since we replace variables, we do not create any more cases that need to checked. Furthermore, when we replace a variable,
	we replace all the variable's repetitions in a literal.
\end{proof}

These definitions are needed for the universal loop. The following are for the existential.

\begin{definition}[Inseparable Clause]
	An \emph{inseparable clause} $K$ is either a singleton literal or requires that there is some variable that is present in all literals of $K$. In the \fragment, this is the same as every pair of literals sharing a variable. 
	
\end{definition}

\begin{definition}[Component]\label{def:sep}
	A \emph{component}\footnote{We borrow this notion from AVATAR \cite{DBLP:conf/cav/Voronkov14}, where components are introduced as the maximal non-splittable subclauses.} $K$ of $C$ is a non-empty inseparable subclause, where if literal $a$ occurs in $K$ then if there is any literal $b$ that occurs in $C$ that shares a variable with $a$ then $b\in K$.
\end{definition}

\begin{example}
	
	In clause $a (u,v,w) \vee b (u,v) \vee p(x)\vee q()$, the components are \sloppy $a (u,v,w) \vee b (u,v)$, $p(x)$ and $q()$.
	For example, subclause $a (u,v,w) \vee b (u,v)$ is an inseparable clause that has outer variables $u$ and $v$. 
	
	$\forall u ,v,w,x \enspace a (u,v,w) \vee b (u,v) \vee p(x)\vee q()$ carries the same meaning as the miniscoped formula 
	$(\forall u ,v,w  \enspace a(u,v,w) \vee b (u,v) )\vee (\forall x \enspace p(x))\vee q()$.
\end{example}

\begin{lemma}\label{lem:mis}
	Given a CNF $\phi\wedge C$ in \fragment with a clause $C$. If $\phi\wedge C$ is satisfiable, then there is some component $K$ of $C$ such that $\phi\wedge K$ is satisfiable.
\end{lemma}

\begin{proof}
	We rely on the following two propositions:
  \begin{enumerate}[nosep]
		\item Two distinct components of a clause do not share
		variables\label{lem:mis1}.
		\item A clause in \fragment is the disjunction of its
		components\label{lem:mis2}.
	\end{enumerate}
	
	Proposition~\ref{lem:mis1} is true because if we have two inseparable subclauses $K_1$ and $K_2$, internally all literals share a variable, and then when $l_1\in K_1$ shares a variable with $l_2\in K_2$ it means $l_1\in K_2$ and $l_2\in K_1$, but then by the QEALM conditions, there will be enough shared prefixes that force that
	all literals of $K_1$ must be in $K_2$ and vice versa.
	Proposition~\ref{lem:mis2} is true because singletons are inseparable subclauses and if we take the closure of literal sets under the operation of including a literal that shares a variable, the first variable argument to appear is always the same, hence the maximal closure of each literal is an inseparable subclause.
	
	Now suppose $\phi\wedge C$ is satisfiable, then we take the interpretation of the predicate symbols that satisfies it. 
	Now suppose, for contradiction, that for each component $K$ of $C$, there is some universal assignment $\alpha_K$ that falsifies it under our interpretation. Then we can restrict $\alpha_K$ to just the variables of $K$ to get $\beta_K$. The  union of all $\beta_K$ over every component $K$ will be a disjoint union and must falsify every $K$, and hence falsifies the clause $C$. By contradiction there must be at least one component $K$ that is always satisfied. Therefore $\phi\wedge K$ is satisfiable.
	
\end{proof}

We present Algorithm~\ref{alg:solvefrag}, which decides the satisfiability of the \fragment.

\begin{algorithm}[h]
	\DontPrintSemicolon
	\SetAlgoNoEnd
	\SetKwProg{Fn}{Function}{:}{}
	
	\Fn{\solve($\phi$)}{
		
		\If{$\phi$ is quantifier-free}{
			\Return \texttt{Sat}($\phi$)\;
		}
		
		\eIf{all clauses are inseparable}{
			\tcp{Universal search}
			
			\For{$\sigma \leftarrow$ total assignments at $\opos(\phi)$}{
				$\phi|_{\sigma} \gets \bigwedge_{C\in\phi} C\sigma$\tcp*{instantiate $\phi$}
				
				\If{\solve($\phi|_{\sigma}$) = UNSAT}{
					\Return UNSAT\;
				}
			}
			\Return SAT\;
			
		}{
			\tcp{Existential search}
			
			\For{$\tau \leftarrow$ one component per clause}{
				$\phi|_{\tau} \gets \bigwedge_{C\in\tau} C$\tcp*{select one component per clause}
				
				\If{\solve($\phi|_{\tau}$) = SAT}{
					\Return SAT\;
				}
			}
			\Return UNSAT\;
		}
		
	}
	\caption{Recursive algorithm for satisfiability of the \fragment \label{alg:solvefrag}}
\end{algorithm}

\begin{example}\label{ex:game1}
	We can take the \bs problem from Example~\ref{ex:game1pre} with constants $c_1,c_2$. It is safe to use the universal quantification on all clauses individually. 
	We can imagine Algorithm~\ref{alg:solvefrag} running like a game between a universal and existential player. Whose turn it is, depends whether all clauses are inseparable or not.
	$$\begin{aligned}
		\forall u_1, u_2, u_3, u_4 \quad\big(y(u_1,u_2,u_3,u_4)\big)\wedge\\ 
		\forall u_1, u_3, u'_3, u_4 \quad\big(\lnot x(u_1,c_1,u_3,c_1)\vee \lnot y(u_1,c_1,u_3,u_4) \vee \lnot x(u_1,c_1,u'_3,u_1)\big) \wedge \\
		\forall u_1, u_3, u_4 \quad\big( x(u_1,u_1,u_3,c_1)\vee \lnot y(u_1,u_1,u_3,u_4) \vee x(u_1,u_1,c_2,u_1)\big)
	\end{aligned}$$
	
	If we follow Algorithm~\ref{alg:solvefrag}, we first observe that all clauses are inseparable and that the first and second position always contain the same argument within clauses.
	Therefore the universal player chooses the value of the variables appearing on the first two arguments. 
	Technically, the various copies of $u_1$ across clauses could be different variables since each clause may as well have its own quantifiers, but it never helps the universal player to assign them differently when they are outer variables (see Observation~\ref{lem:semantic}). The meaning of the outer variables is that we can take them outside of the formula:
	$$\begin{aligned}
		\forall u_1, u_2
		\big(\\
		\forall u_3, u_4 \quad\big(y(u_1,u_2,u_3,u_4)\big)\wedge\\ 
		\forall u_3, u'_3, u_4 \quad\big(\lnot x(u_1,c_1,u_3,c_1)\vee \lnot y(u_1,c_1,u_3,u_4) \vee \lnot x(u_1,c_1,u'_3,u_1)\big) \wedge \\\ 
		\forall u_3, u_4 \quad\big( x(u_1,u_1,u_3,c_1)\vee \lnot y(u_1,u_1,u_3,u_4) \vee x(u_1,u_1,c_2,u_1)\big)\big)
	\end{aligned}$$
	
	For a universal assignment we assign both the first and second argument ($u_1$ and $u_2$) to $c_1$ and the second argument. %
	Each application of a universal assignment yields another problem still in the fragment.  
	$$\begin{aligned}
		\forall u_3, u_4 \quad\big(y(c_1,c_1,u_3,u_4)\big)\wedge\\ 
		\forall u_3, u'_3, u_4 \quad\big(\lnot x(c_1,c_1,u_3,c_1)\vee \lnot y(c_1,c_1,u_3,u_4) \vee \lnot x(c_1,c_1,u'_3,c_1)\big) \wedge \\\ 
		\forall u_3, u_4 \quad\big( x(c_1,c_1,u_3,c_1)\vee \lnot y(c_1,c_1,u_3,u_4) \vee x(c_1,c_1,c_2,c_1)\big)
	\end{aligned}$$
	
	At this point some clauses are separable into components, this means we can rewrite the problem using miniscoping.
	$$\begin{aligned}
		\big(\forall u_3, u_4 \enspace y(c_1,c_1,u_3,u_4)\big)\wedge\\ 
		\big(\forall u_3, u_4 \enspace \lnot x(c_1,c_1,u_3,c_1)\vee \lnot y(c_1,c_1,u_3,u_4) \big) \textcolor{red}{\vee} \big(\forall u'_3\enspace\lnot x(c_1,c_1,u'_3,c_1)\big) \wedge \\\ 
		\big(\forall u_3, u_4 \enspace x(c_1,c_1,u_3,c_1)\vee \lnot y(c_1,c_1,u_3,u_4)\big) \textcolor{red}{\vee}\big( x(c_1,c_1,c_2,c_1)\big)
	\end{aligned}$$
	
	Now it is fine for the existential player to choose the disjuncts he will use to gain a satisfiable ground problem. There are $1 \times 2 \times 2$ options this time.
	After choosing the disjunct we remain in the fragment, and now return to every clause being inseparable.
	However unfortunately for the existential player no matter the choice here, if we continue to play this game the universal player can force a contradiction. 
	
\end{example}

One alternative formulation is to understand components as abstracted to propositional literals and we loop through assignments to those variables during the existential search. This formulation may lead to more interesting solving approaches (i.e. clause learning or CEGAR: counter-example guided abstraction and refinement), but we do not need it here to show correctness and constraint learning complicates our incoming argument about space complexity.

\begin{lemma}[Correctness]\label{lem:correctness}
	If $\phi$ is a \fragment problem that does not contain the empty clause then 
	$\solve(\phi)$ (Algorithm~\ref{alg:solvefrag}) returns SAT if and only if there is a satisfying interpretation for $\phi$.
\end{lemma}

\begin{proof}
	We can perform induction on the number of quantified variables + predicate symbols appearing in the formula
	
	\noindent\textbf{Base Case: Fully Ground.}
	If $\phi$ is ground and \solve($\phi$) returns SAT then there is a satisfying Boolean assignment to the occurring ground literals. This can be used to construct the truth table for the predicates, and these predicates will satisfy $\phi$.
	If $\phi$ is ground and \solve($\phi$) returns UNSAT then there is no satisfying Boolean assignment to the ground values, hence there is no truth table interpretation of the predicates which does not falsify $\phi$, and a predicate interpretation needs a truth table to exist, so there is no satisfying interpretation.
	
	\noindent\textbf{Inductive Step: Universal Search.}
	If $\phi$ has a satisfying interpretation then the same satisfying interpretation satisfies all clauses under instantiation, therefore $\solve(\phi)$ returns SAT after it exhausts all assignments to the variables at $\opos(\phi)$.
	If $\phi$ has no satisfying assignment then it has a Ground Resolution refutation, however said refutation cannot mix variables from different instantiations of the fully unified argument prefix (there are no resolution paths see Obs~\ref{lem:semantic}). Therefore, there is some instantiation of the $S$ variables where $\phi$ has no satisfying interpretation and hence \solve($\phi$) returns UNSAT.
	
	\noindent\textbf{Inductive Step: Existential Search.} If $\phi$ is unsatisfiable then any problem composed of subclauses of $\phi$ is also unsatisfiable, therefore \sloppy $\solve(\phi)$ returns UNSAT.
	Now suppose $\phi$ has a satisfying interpretation.
	We can take $\phi$ out of its prefix form.
	Since $\forall$ and $\wedge $ commute, each clause could have its own quantification, then we push all universal quantifiers to the level of the components.
	Lemma~\ref{lem:mis} is important because components are disjoint and quantified variables are not shared between them. 
	So we have a disjunction where each disjunct is universally quantified. The satisfying interpretation must satisfy at least one disjunct in each conjunct, and every assignment of disjuncts gets selected in the loop.
\end{proof}

\begin{lemma}[\PSPACE\xspace inclusion]\label{lem:memory}
	Given a suitable SAT solver black box, the amount of memory of Algorithm~\ref{alg:solvefrag} is bounded from above by a polynomial in the size of the input CNF $\phi$, if $\phi$ is a \fragment problem.
\end{lemma}
\begin{proof}
	The algorithm requires a stack of recursive calls. The stack size of recursive calls from a universal search is bounded by the largest arity of a predicate or the number of universal quantifiers, since each instantiation sets a constant value to at least one variable that is shared in all literals.
	The stack of the recursive calls from an existential search is bounded by the number of times we can split at least one clause, so an upper bound is the total number of literals.
	Therefore there are at most a linear number of recursive calls at any one time. Moreover, each recursive call only uses a polynomial amount of memory for the loops, propositional logic and the SAT call. Hence the algorithm uses \PSPACE. 
\end{proof}
This algorithm essentially saves memory because Observation~\ref{lem:semantic} allows us to query the outer variables of different clauses jointly. The QEALM condition allows us to recursively remain in the position where we can either apply miniscoping or Observation~\ref{lem:semantic} and afterwards remain in QEALM. Without the QEALM conditions, we may eventually arrive at a position where we can apply neither.
In fact, with a bit more detail %
we can do better than just \PSPACE\xspace and identify the fragments that belong in classes of the polynomial hierarchy. The idea  is to identify which argument positions may force an existential search in Algorithm~\ref{alg:solvefrag}. If we bound the number of times an existential search happens between universal queries we bound the number of alternations by roughly twice of that. 

\begin{definition}
	Let $\Phi$ be a problem in the \fragment.
	Position $i>0$ is a \emph{fork index} if there is a clause $C \in \Phi$ and a pair of literals $l_1$ and $l_2$ in $C$ (not necessarily distinct)
	such that $i$ is the minimal value such that all variables shared between $l_1$ and $l_2$ occur as arguments in positions $\leq i$.
\end{definition}

\begin{example}Let $u_1 \dots $ denote universal variables, $c_1 \dots $denote constant symbols and $p_1 \dots $ denote predicate symbols.
	For $\forall u_1 u_2 u_3\big(p_1(u_1,u_2) \vee p_2(u_1,u_2, c_1,u_3)\big) \wedge \big(\neg p_2(u_1, u_2, u_1, u_3\big)$.	
	The fork indices are $2$ and $4$. 
\end{example}

After a universal search the least
current fork index is removed because joint variables are instantiated, whereas an existential search will result in no new fork index.

\begin{lemma} [$\Sigmap{2k-1}$, $\Pip{2k}$ \xspace inclusion]\label{lem:ph}
	Given a set $\{\Phi_j\}$ of problems in the \sloppy \fragment, let the number of fork indices be bound by $k$.
	Then $\{\Phi_j\}$ is a family in the complexity class:
	\begin{itemize}
		\item $\Pip{n}$ if each formula of $\{\Phi_j\}$ has only inseparable clauses, where $n=2k$.
		\item $\Sigmap{n}$ in all $\{\Phi_j\}$, where $n=2k+1$, otherwise
	\end{itemize}
\end{lemma}

\begin{proof}%
	By induction on the maximum number of recursive calls of Algorithm~\ref{alg:solvefrag} would need to solve it, which is a finite value by Lemma~\ref{lem:memory}.
	
	\noindent\textbf{Base Case:} When there is no more recursive calls, there are no quantified variables and no fork indices. Thus we can solve the problem as a $\Sigma^P_1$ problem with SAT solving.
	If all clauses are inseparable then there are no disjunctions and we fall into \textsf{PTIME}=$\Pi^P_0$.
	
	\noindent\textbf{Inductive Step ($\Pi^P_{2k}$ case):} Suppose all clauses, of the matrix $\phi$ are inseparable.
	Assume we are not in the base case, and there is at least one variable, so there is some minimum fork index $m>0$.
	For every clause $C\in \phi$ and every $0< j\leq m$ must be a prefix position, otherwise due to the definition of QEALM there is a fork index strictly less than $j$. Suppose we have some literal $l\in C$ that contains a variable $v$ at argument position $0< j\leq m$, since this is a prefix position it must either be an outer position for $C$ or be some repetition of a $v$ that occurred first at argument $i\leq j$.
	Whatever the case, during the universal search in Algorithm~\ref{alg:solvefrag}, $\sigma$ will assign this variable to a constant. Therefore all arguments up to and including $m$ for all literals in all clauses in $\bigwedge_{C\in\phi} C\sigma$ will not be variables. 
	Therefore $m$ is no longer a fork index in $\bigwedge_{C\in\phi} C\sigma$.

	Suppose $p>0$ is a fork index in some clause $C\sigma$ then we have a pair of literals $l_1\sigma$ and $l_2\sigma$ that share all common variables minimally after the $p$th argument. But in that case there is some variable $u$ occurring at $p$ that is not assigned by $\sigma$, $u$ must be present in $l_1$ and $l_2$. The only way $p$ cannot be a fork index in $C$ is if there is some variable $v$ shared between $l_1$ and $l_2$ that first occurs after $p$ at argument position $q$ that is assigned by $\sigma$. That means $q$ is a prefix position in all clauses, so $p$ must be prefix in all clause. $q$ must be outer in $C$. So $p$ must also be outer in $C$, so $p$ would not be a fork index after instantiation as presumed. 
	
	Thus we decrease the number of forks indices by $1$, but now have clauses that are not inseparable therefore $\bigwedge_{C\in\phi} C\sigma$ is always in $\Sigma^P_{2k+1-2}$ via the induction hypothesis.
	Since, we query over a polynomial size  $\sigma$ to find a $\Sigma^P_{k+1}$ to falsify, deciding $\phi$ is in $\Pi^P_{2k}$.
	
	\noindent\textbf{Induction Hypothesis ($\Sigma^P_{2k+1}$ case):}
	Suppose we have a number of separable clause. In the existential search of Algorithm~\ref{alg:solvefrag}, we filter this to selecting individual components, all clauses in $\bigwedge_{C\in \tau} C$ are inseparable, but we create no new fork indices as each clause is a subset of version in $\psi$, and so we consider a subset of the set of pairs of literals to find the fork indices. Therefore $\bigwedge_{C\in \tau} C$ is a $\Pi^P_{2k}$ problem. Since we loop over the choice of (polynomial size) $\tau$ and accept when one selection works we are in $\Sigma^P_{2k+1}$.
\end{proof}

We can take inclusion further and show completeness; first in the \PSPACE{} case and then for the classes of the polynomial hierarchy.

\begin{lemma}[\PSPACE-hardness]\label{lem:hardness}
	Any closed (totally quantified) QBF problem can be transformed into a \fragment problem with polynomial time computable function $f$,
	where $f(\Psi)$ is satisfiable if and only the QBF $\Psi$ is true.
\end{lemma}

\begin{proof}[Sketch Proof]
	For $f$, we use the transformation detailed in Section~\ref{sec:qbf2epr} up to the end including the removal of the first order interpretation of universal literals. 
	We only consider prenex QBFs in CNF (without tautological clauses) but this will be sufficient for \PSPACE-hardness. Note here that our domain is QBFs not DQBF, so the dependency set of existential variables are those universal variables left of it.
	A prenex QBF is true if and only if it has a set of satisfying Skolem functions, which is precisely what is encoded in the transformation.
	
\end{proof}

\begin{theorem}\label{thm:completeness}
	The \fragment is \PSPACE-complete.
\end{theorem}
\begin{proof}
	For hardness we have the transformation in Lemma~\ref{lem:hardness}. 
	For inclusion in \PSPACE\xspace we use Algorithm~\ref{alg:solvefrag},  except where there is an empty clause then the formula is false. The properties of Lemma~\ref{lem:correctness}~and~\ref{lem:memory} complete the argument.
\end{proof}

\begin{theorem} [$\Sigmap{2k-1}$, $\Pip{2k}$-completeness]
	\label{thm:phcompleteness}
	The following is true:
	\begin{itemize}
		\item The set of all \fragment problems which have $0$ as a fork index and $k$ or fewer fork indices in total is $\Sigmap{2k-1}$-complete.
		\item The set of all \fragment problems which have $k$ or fewer fork indices in total is $\Pip{2k}$-complete.
	\end{itemize}
\end{theorem}
\begin{proof}
	Using Lemma~\ref{lem:ph} we get inclusion.
	Hardness can be shown by transforming a Prenex QBF with a bounded number of alternations to \fragment and while staying within the boundaries of side conditions in the Theorem statement.
	
	Suppose we have a Prenex QBF which starts with an existential quantifier and then has $2k-1$ quantifier blocks $\exists_B X_1 \forall_B U_2 \exists_B X_2 ... \exists_B X_k$ (this is $\Sigmap{2k-1}$-complete) or we start with a universal quantifier and have $2k$ quantifier blocks $\forall_B U_1 \exists_B X_1 \forall_B U_2 \exists_B X_2 ... \exists_B X_k$ (this is $\Pip{2k}$-complete).
	In either case let it end with an existential block. After the prefix, there is the propositional part known as the matrix. We can use a Tseitin transformation to change it into conjunctive normal form.
	The transformation introduces new existentially quantified variables that are to be placed in the final existential block after all other variables. This preserves satisfiability, as any winning strategies can be extended to play the new variables according to their definitions. Therefore, after the Tseitin transformation, we retain the same number of quantifier blocks.
	
	We can then use the conversion to the \fragment as given in \sloppy Section~\ref{sec:qbf2epr},
	after which we observe the following: $0$ is a fork index if and only if the QBF begins with an existential quantifier and $i>0$ is a fork index if and only if there is some $m>0$
	such that the sum $\sum^{m}_{j=1} |U_j|=i$.
	Hence there are going to be $k$ fork indices in total and the condition on the $0$-th index is also respected.
\end{proof}

The above hardness results, use the accessible transformation using \texttt{qbf2epr}, other translations can be used to prove stronger results. 
Namely, that the intersection of \krom, \horn, \affine and the \fragment remains \PSPACE-complete.

\begin{theorem}\label{thm:intersection}
	The \fragment is still \PSPACE-hard with the following restrictions.
  \begin{enumerate}[nosep]
		\item Every clause has size at most 2 (is Krom).
		\item Every clause has at most one positive literal (is Horn).
		\item For every 2-element clause $a \rightarrow b$, also the opposite implication $b \rightarrow a$ is present among the clauses (is Affine).
	\end{enumerate}
\end{theorem}
\ifdefined\Fversion
\begin{proof}
	See Appendix~\ref{app:proof} for a full proof.
\end{proof}
\fi
On the other hand we can relax the ordering found the \fragment, to categorise a wider range of \PSPACE\ problems that may be of practical interest. 

\begin{definition}[QEALM re-arrangeable]\label{def:rearrange}
	A \bs problem $\phi$ is \emph{QEALM re-arrangeable} if for every predicate symbol $p$ there is a permutation $s_p: [\arity(p)]\rightarrow [\arity(p)]$, such that if we rearrange the order of arguments of every occurrence of every $p$ and do so for every predicate, then we have a QEALM problem.
\end{definition}

QEALM re-arrageability is a similar notion to HORN-renameability. The satisfiability of QEALM re-arrangeability is also \PSPACE-complete, because we can enumerate all possible symmetries with one additional alternation. While the detection of whether an instance is in QEALM is in polynomial time in the worst case, the detection of QEALM re-arrangeable problems may not be. But we would hope in practice there will be many easily detectable cases.

Finally, on the theoretical side, we present an important first-order closure property.%
\begin{theorem}
	The \fragment is closed under the rules of Robinson's resolution. \label{thm:closure}
\end{theorem}

\begin{proof}
	
	We can argue that Robinson's rules Resolution and Factoring preserve this class. 
	
	\noindent\textbf{Resolution.}
	Let us use uppercase $C,D,E$ for clauses,  $p,q,r$ for predicate symbols and lowercase greek letters for argument sets. 
	Suppose $C\vee p(\alpha)$ resolves with $D\vee \lnot p(\beta)$ to get $E$, let us take two literals $q(\gamma)$ and $r(\delta)$ in $E$ and suppose they have some shared variable $u$, the first occurrence of $u$ in $q$ is argument number $i$, and in $r$ is $j$.
	This can only occur if these first $u$ occurrences
	1) were equal before the resolution, or 
	2) unified as a result of the resolution step. 
	
	Before instantiation $q(\gamma)$ and $r(\delta)$ exist as literals $q(\epsilon)$ and $r(\zeta)$.
	We generally assume that two clauses share no variables between them until after unification. So case 1 only occurs when the $q(\epsilon)$ and $r(\zeta)$ already appear in the same clause and so $i=j$ and all variables $<i$ are already the same by QEALM rules. The further instantiation that happens inside the resolution step cannot make these variables different.
	In the second case the unification occurs because of the instantiation of the two clauses, the clause $C\vee p(\alpha)$ is instantiated to effect the variables of $\alpha$ and all other variables in $C$ that are equal to those in $\alpha$, likewise with $D$ and $\beta$.
	
	We then argue in case 2  that $q(\epsilon)$ and $r(\zeta)$ are not from the same clause. WLOG we assume that they both are in $C$ and  $i\leq j$. 
	Let $v$ be the $i$-th argument of $\epsilon$ and $w$ be the $j$-th argument of $\zeta$, at least one of these has to change as a result of resolution to match the other, but this only occurs if both  $v$ and $w$ occur in $\alpha$ otherwise the unifier is not the most general. %
	Then by the definition of \fragment and our assumption that $i\leq j$, $v$ must occur in $\zeta$ and thus we are in Case 1.
	So without loss of generality $q(\epsilon)$ is in $C$ and $r(\zeta)$ is in $D$ and $\epsilon$ has a variable $v$ which also occurs in $\alpha$ (occurring first in the $i$-th argument) and $\zeta$ has a variable $w$ also occurring in $\beta$ (first in the $j$-th argument) and $v$ and $w$ will become the same variable $u$.
	
	In order for $v$ and $w$ to both become the same variable they need to be part of some branch of variables that share a common argument index across $\alpha$ and $\beta$. Such a path will always unify to the same term, and if a constant term appears in the path it cannot be a variable. 
	If we suppose $i<j$ there is some term $t$ in the $i$-th argument of $\zeta$ (and thus in the $i$-th argument of $\beta$) that does not become $u$. If $t$ is constant then $v$ cannot become $u$ under the unification, so $t$ must be a variable, and it must become $u$, like $v$ under instantiation and therefore $j\leq i$. This works symmetrically and so $j=i$ is the only possibility.
	Therefore for $k<i$ all the arguments of $q(\epsilon)$ in $C$ are the same as the arguments of $p(\alpha)$ and likewise all the $k<i$ arguments of $y(\zeta))$ in $D$ are the same as $\lnot p(\beta)$. Therefore after instantiation these also become equal. 
	
	\textbf{Factoring.} In factoring, the clause is most generally instantiated so that two literals $p(\alpha)$ and $p(\beta)$ under the same predicate $p$ can become identically unified to $p(\gamma)$ . 
	Because of the condition of the fragment, there is some maximal $i$ (possibly zero) such that for $j\leq i$ all variables already are identical and for $j>i$ all variables are different. 
	Therefore the instantiation will only occur on the arguments $j>i$ that correspond to these different variables. 
	Thus for any other literal $y(\delta)$ in the clause they become affected if they share a variable with $p(\alpha)$ (or $p(\beta)$ but WLOG let us assume $x(\alpha)$), in which case the first argument $j>i$ where the variable appears is the same.
	All previous variables  to $j$ in $\delta$ are already the same as in  $\alpha$. Furthermore those $\delta$ variables before argument $i$ already match with $\beta$.
	Any of the $\delta$ variables before $j$ get the identical treatment to $\alpha$ and agree with $\gamma$. 
	This means if any pair of literals that shares a variable $u$ is created from this instantiation they will have already agreed on the variables $k\leq i$ and will be instantiated to equal on the variables after $i$ up to the argument containing $u$.
\end{proof}

\krom and \horn fragments of \bs are also closed under Robinson's resolution. 
QEALM re-arrangeability is also closed under Robinson's resolution as the rules commute with the application of global symmetries on the arguments. 
Note that the translation of QBF in Lemma~\ref{lem:hardness} can not be closed under Robinson's resolution, because resolution steps produce formulas in QEALM that have no preimage in the QBF translation. The original  translation  \texttt{qbf2epr} is also not closed under Robinson's resolution.

\section{Software and Experimentation}

\subsection{Implementing a QEALM Detector} \label{subsect:detect}
In the TPTP library~\cite{Sut17} version 9.1.0, in its CNF subset, there are {936} \bs problems.
We wrote a small extension of the Vampire theorem prover~\cite{DBLP:conf/cav/BartekBCHHHKRRRSSV25} (mainly, to reuse its efficient parser)
to test for membership in the \fragment{} and establish the corresponding number of fork indices.\footnote{Available at \url{https://github.com/vprover/vampire/tree/martin-epr-fragment}.}

We discovered {308} \fragment{} problems of which 1)
50 are purely propositional (\# fork indices $\leq$ 1),
109 monadic (\# fork indices = 2), and
49 contain predicates with up to two ``variable-active'' argument positions (\# fork indices = 3).\footnote{
	For example, the problem \texttt{MGT041-2} from this group contains a predicate %
	of arity 3, but there a variable never occurs at its third argument position.}
Many of these last named 49 problems are ALC problems translated from propositional multi-modal K
logic formulae generated according to the scheme described by Hustadt and Schmidt \cite{DBLP:conf/ijcai/HustadtS97}
(e.g., the problem \texttt{SYN521-1}).
The problems with higher (in principle unbounded) number of fork indices are from the following two families:
\begin{itemize}
	\item
	binary counter formulas by P{\'{e}}rez and Voronkov~\cite{DBLP:conf/birthday/PerezV13} (which are also \krom),
	e.g., the problem \texttt{MSC015-1.005}, and
	\item
	optimised functional translations of modal QBF formulae by Hustadt and Schmidt~\cite{DBLP:conf/tableaux/HustadtS00},
	e.g., the problem \texttt{SYN870-1}.
\end{itemize}

It is possible that more thorough preprocessing of the input would detect more examples. We only used the default argument ordering of the predicates.
We also remark that our procedure does not detect 81 \bs problems translated to TPTP by \texttt{qbf2epr}~\cite{SLB12}.
Strictly speaking, they do not belong to the \fragment{}, although by eliminating the
two---always present---clauses $p(\mathit{true})$, $\neg p(\mathit{false})$, by exhaustively resolving against all other $p$-literals as described in Sect.~\ref{sec:qbf2epr}
(and simply dropping the superfluous clause $\mathit{true} \neq \mathit{false}$), the membership could be recovered.

\subsection{Solving QEALM with Vampire}

To probe the landscape of the discovered QEALM problems in TPTP (cf.~Sect.~\ref{subsect:detect}), we tried to solve each of them using Vampire.
Here are our two main findings:
\begin{itemize}
\item
	While Vampire's default strategy exhibits exponential behavior on the binary counter formulas and produces the obvious ground proof which ``actually does the counting'' (by deriving exponentially many unit clauses),
	by changing the ordering and literal selection function\footnote{Choosing \texttt{--kbo\_weight\_scheme precedence} and a size-preferring \texttt{--selection 2} (see \cite{DBLP:conf/cade/HoderR0V16}).}
	the prover can be nudged to run in what appears to be polynomial time (roughly $\Theta(k^{3.5})$, where $k$ is the number of the counter's bits)
	 and to produce non-ground proofs which avoid the exponential worst-case.

\item
	On the translations of modal QBF formulae, the number of fork indices is a good predictor of
	the effort needed to solve the formula by saturation (unlike $C$, the number of clauses in the source QBF in the TPTP header):
Over the $90$ of $92$ such problems that Vampire solved under a fixed budget of 
$1.28\cdot 10^{11}$ instructions, the Spearman rank correlation between the
number of fork indices and the logarithm of the instructions burned is
$\rho = +0.76$ ($p < 10^{-17}$), the median cost rising from $548\cdot 10^{6}$
instructions at $21$ fork indices to $25\,577\cdot 10^{6}$ at $57$, whereas for
$C$ it is $\rho = -0.05$ ($p = 0.64$).
\ifdefined\Fversion
\footnote{See Appendix~\ref{app:plot} for a corresponding plot.}
\fi
\end{itemize}

Beyond these two findings, the space behaviour of the prover deserves a remark.
Across the whole collection, Vampire's peak memory stayed in the tens of
megabytes on the large majority of the problems, and it grew markedly more
slowly than the running time: over the range of fork indices spanned by the QBF
translations, the instructions burned grow by an order of magnitude more than
the memory consumed. We do not offer this as empirical evidence of the
\PSPACE\ bound---saturation carries no space guarantee, and the runs that
exhausted the instruction budget did so while occupying hundreds of megabytes---
but it does suggest that on \fragment{} problems the clause sets that a
saturation-based prover must keep in memory remain modest, and that time rather
than space is the binding resource in practice.
\ifdefined\Fversion
\footnote{See
Appendix~\ref{app:memory} for more details.}
\fi

Currently, there are no dedicated practical QEALM solvers. Algorithm~\ref{alg:solvefrag} presents a \PSPACE\ approach, but lacks any heuristics and optimisations necessary to compete with Vampire.
Perhaps one positive indication that we can hope to solve QEALM problems more practically, is the closeness to QBF. 
There are of course a number of directions that can be taken in future steps, all of which can be inspired by QBF solving. 

\section{Conclusion}
This work  widens our understanding of first-order logic by finding a new \sloppy \PSPACE-complete fragment, that shares many similarities to QBF. 
Remarkably, 
despite not being defined on propositional connectives, the \fragment has multiple properties shared by \krom and \horn. QEALM instances can be detected per-clause,  QEALM is closed under Robinson's resolution, and retains hardness  when intersected with \krom, \horn and \affine  (unlike QBF~\cite{Buning1995,Remshagen}). %

Part of our motivation is to utilise the \PSPACE-completeness and connection to QBF for benefits in solving, certification and theory.
At present our recursive solving algorithms shows \PSPACE\xspace inclusion, but lacks 
optimisations we would need to use it practically. 

One area of future interest that can also take advantage of is QBF proof complexity, as our understanding of it has grown substantially over the last decade. 
For example, some QBF proof systems are known to have super-polynomial lower bounds only if either proof complexity or circuit complexity open problems are solved~\cite{BBCP20}. In theory, a proof system with this property exists in every \PSPACE\xspace language through a QBF translation, but finding one that is natural and suitable for applications is more difficult. %
Unfortunately, Robinson's resolution is known to have exponential lower bounds even for the \fragment. In fact, it runs foul of something similar to a well-known QBF proof complexity pitfall known as ``semantic hardness''~\cite{BBH19}, which we aim to demonstrate in future work.

\subsection*{Acknowledgements}
This work was supported by the European Union under the project ROBOPROX (reg.\
no.\ CZ.02.01.01/00/22\_008/0004590) and by the Czech Science Foundation project 24-12759S.

\label{sect:bib}
\bibliographystyle{plainurl}

\newpage

\ifdefined\Fversion
\appendix
\section{Proofs from Section~\ref{sec:frag}}

\label{app:proof}

\begin{theorem}
	The \fragment is still \PSPACE-hard with the following restrictions.
	\begin{enumerate}
		\item Every clause has size at most 2.
		\item Every clause is Horn.
		\item For every 2-element clause $a \rightarrow b$, also the opposite implication $b \rightarrow a$ is present among the clauses.
	\end{enumerate}
\end{theorem}
\begin{proof}
	We describe a polynomial reduction of QBF to this fragment of \bs. 
	Consider a QBF
	$$
	\exists_B x_1 \forall_B y_1 \exists_B x_2 \forall_B y_2 \ldots \exists_B x_n \forall_B y_n
	\cnf(x_1, y_1, \ldots, x_n, y_n),
	$$
	where $\cnf$ is a CNF-formula. For an initial list of variables $x_1, \ldots, y_k$, let us denote
	$\qbf(x_1, \ldots, y_k)$ the value of the formula above with the first $2k$ variables fixed.
	So $\qbf()$ is the initial QBF instance, and $\qbf(x_1, \ldots, y_n) = \phi(x_1, \ldots, y_n)$.
	
	We construct the equivalent \bs problem using a single $(2n+1)$-ary predicate $p$,
	and the following constraints:
	
	\begin{enumerate}
		\item 
		For each clause $C$ in $\phi$ we write the equivalence
		$$p(t^C_{x_1},t^C_{y_1},t^C_{x_2},t^C_{y_2},\ldots,t^C_{x_n},t^C_{y_n},0) \leftrightarrow p(t^C_{x_1},t^C_{y_1},t^C_{x_2},t^C_{y_2},\ldots,t^C_{x_n},t^C_{y_n},1)$$
		
		Where:
		
		$$t^C_{v}= \begin{cases} 0 & \text{ if positive literal } v \in C \\
			1 & \text{ if negative literal } \lnot v \in C  \\
			v & \text{  otherwise }
		\end{cases}$$
		
		So for example, if the clause is $y_1\vee x_2 \vee \neg y_1$, then we write
		$$p(x_1,0,0,1,x_3\ldots,x_n,y_n,0) \leftrightarrow p(x_1,0,0,1,x_3\ldots,x_n,y_n,1)$$
		\item For each $k = 1, \ldots, n$, we write the equivalence
		$$p(x_1,y_1,\ldots,x_{k-1},y_{k-1},1,y_k,0,y_{k+1},\ldots,0,y_n,0) \leftrightarrow$$
		$$p(x_1,y_1,\ldots,x_{k-1},y_{k-1},0,z_k,1,z_{k+1},\ldots,1,z_n,1) $$
		\item A positive singleton $p(0,y_1,0,y_2,\ldots,0,y_n,0)$
		\item A negative singleton $\neg p(1,y_1,1,y_2,\ldots,1,y_n,1)$
	\end{enumerate}
	
	To prove that this \bs formula is equivalent to $\qbf()$, we frame this \bs formula as a reachability problem.
	The \bs formula is contradictory if and only if there exists values $y_1, \ldots, y_n$, $ z_1, \ldots, z_n$, and a path on ground instances which:
	\begin{itemize}
		\item starts at $p(0,y_1,0,y_2,\ldots,0,y_n,0)$
		\item ends at $p(1,z_1,1,z_2,\ldots,1,z_n,1)$, and
		\item takes steps along the equivalences of type (1), and (2).
	\end{itemize}
	
	We prove a more general claim by induction. Fix $k = 0, \ldots, n$, and values for $x_1, y_1, \ldots, x_k, y_k$.
	Then $\qbf(x_1, \ldots, y_n)$ is false if and only if there is a path on ground
	instances of the formula such that:
	\begin{itemize}
		\item it starts at $p(x_1,y_1,x_2,y_2,\ldots,x_k,y_k,0,?,0,?,\ldots,0,?,0)$ for some values of the question marks.
		\item it ends at $p(x_1,y_1,x_2,y_2,\ldots,x_k,y_k,1,?,1,?,\ldots,1,?,1)$ for some values of the question marks.
		\item takes steps along the equivalences of type (1), and (2) which do not change the first $2k$ positions $x_1,\ldots,y_k$.
	\end{itemize}
	
	We proceed by induction for $k = n, n-1, \ldots, 1, 0$. If $k = n$, only rule (1) s applicable, and it is applicable
	if and only if there is a clause in $\cnf$ which is false about $x_1, \ldots, y_n$, hence the claim follows by definition.
	
	Now, assume $k < n$, and that the claim is true for $k+1$. If $\qbf(x_1, \ldots, y_k)$ is false, then
	there exists independent universal responses $y_{k+1}$, and $z_{k+1}$ such that both
	$$
	\qbf(x_1, \ldots, y_k, 0, y_{k+1}),
	\qbf(x_1, \ldots, y_k, 1, z_{k+1})
	$$
	are false as well.
	
	By induction assumption, there are paths
	$$
	p(x_1, \ldots, y_k, 0, y_{k+1}, 0, y_{k+1}, \ldots, 0, y_n, 0) \leftrightarrow \cdots
	$$$$
	\leftrightarrow p(x_1, \ldots, y_k, 0, y_{k+1}, 1, y'_{k+1}, \ldots, 1, y'_n, 1)
	$$$$
	\text{and}
	$$$$
	p(x_1, \ldots, y_k, 1, z_{k+1}, 0, z'_{k+1}, \ldots, 0, z'_n, 0) \leftrightarrow \cdots
	$$$$
	\leftrightarrow p(x_1, \ldots, y_k, 1, z_{k+1}, 1, z_{k+1}, \ldots, 1, z_n, 1)
	$$
	for some values of $y_i$, $y'_i$, $z'_i$, $z_i$ ($i > k$). These two paths can be connected with a single
	step (2).
	$$
	p(x_1, \ldots, y_k, 0, y_{k+1}, 1, y'_{k+1}, \ldots, 1, y'_n, 1) 
	$$$$
	\leftrightarrow
	$$$$
	p(x_1, \ldots, y_k, 1, z_{k+1}, 0, z'_{k+1}, \ldots, 0, z'_n, 0)
	$$
	
	For the opposite implication, we assume that there is a path, and want to show that the QBF is not true.
	The only rule that can change the value of $x_{k+1}$ is rule (2) with $k := (k+1)$.
	Without loss of generality, we may assume that it is applied only once in the path changing $x_{k+1}$
	from 0 to 1: if it occurs multiple times, both first and last steps of this kind must change $x_{k+1}$
	from 0 to 1, and we can delete all the steps in between.
	
	Now, the value of $y_{k+1}$ can change at most once as well -- at the moment of $x_{k+1}$ changing.
	We denote its first value as $y_{k+1}$, and its second value as $z_{k+1}$. By induction hypothesis,
	both
	$$
	\qbf(x_1, \ldots, y_k, 0, y_{k+1}), \quad \qbf(x_1, \ldots, y_k, 1, z_{k+1})
	$$
	are false, so as well $\qbf(x_1, \ldots, y_k)$ finishing the induction step, and the proof as well.
\end{proof}

\section{Vampire solving the translated modal QBF formulae}
\label{app:plot}

\begin{figure}[t]
  \centering
  \includegraphics[width=\linewidth]{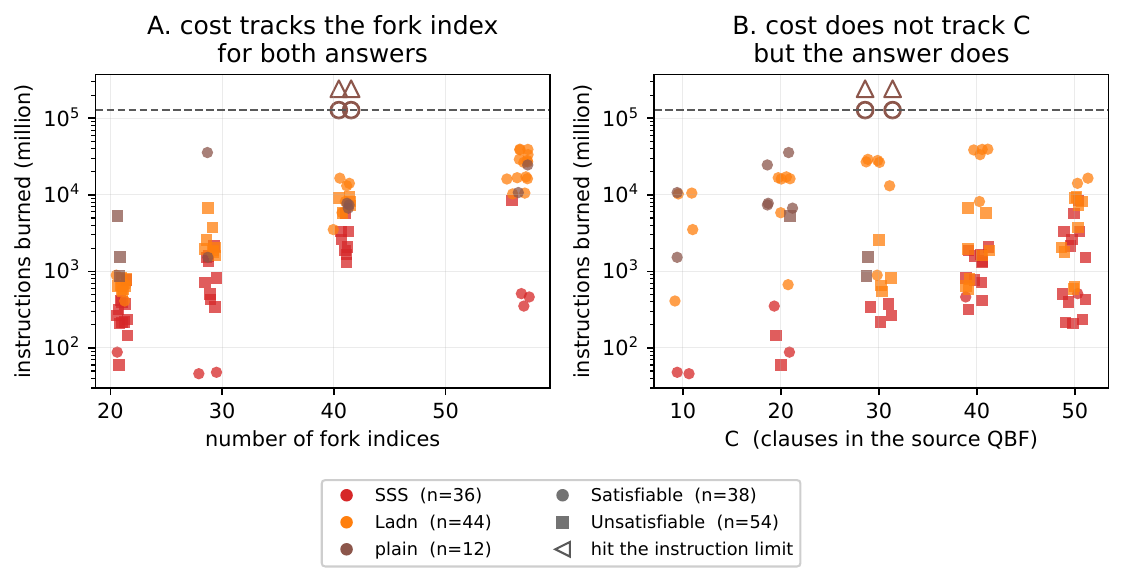}
  \caption{The $92$ optimised functional translations of modal QBF formulae,
    under Vampire's default strategy with an instruction limit of
    $1.28\cdot 10^{11}$. Colour distinguishes the three QBF-to-$K$
    translations, marker shape the answer. Left: cost against the number of
    fork indices. Right: the same problems against $C$. Hollow markers denote
    the two runs that reached the instruction limit, whose cost is therefore
    only a lower bound.}
  \label{fig:qbf}
\end{figure}

Fig.~\ref{fig:qbf} shows that on the optimised functional translations of modal QBF formulae, the number of
fork indices is a good predictor of the effort needed to solve the formula by
saturation (unlike $C$, the number of clauses in the source QBF in the TPTP
header).
The Spearman rank correlation between the number of fork indices and the logarithm of the instructions burned
remains clear when computed separately for the satisfiable and the unsatisfiable instances ($\rho = +0.59$
and $\rho = +0.77$, respectively), whereas $C$---which the generator varies
independently over $10$--$50$---governs satisfiability rather than cost
($\rho = +0.58$ with unsatisfiability, but only $-0.05$ with the logarithm of
the instructions burned).

\section{Memory consumption}\label{app:memory}

Over the $292$ of $308$ problems that Vampire's default strategy solved within
the budget of $1.28\cdot10^{11}$ instructions, the median peak memory was
$14$\,MB; $82\%$ of the runs stayed below $20$\,MB and $94\%$ below
$100$\,MB. Memory grows substantially more slowly than time: on the QBF
translations, the median instruction count rises by a factor of $47$ between
$21$ and $57$ fork indices while the median peak memory rises only by a factor
of $9.5$, and a direct fit over those $90$ problems gives
memory~$\approx$~instructions$^{0.50}$ ($R^2 = 0.85$).

The small footprint is a property of the searches that terminate. The $16$ runs
that reached the instruction limit had a median peak of $888$\,MB and a maximum
of $1235$\,MB, and these are lower bounds, the runs having been cut off. In
Figure~\ref{fig:memory} they are the hollow markers, which sit some two orders
of magnitude above the rest.

\begin{figure}[t]
  \centering
  \includegraphics[width=\linewidth]{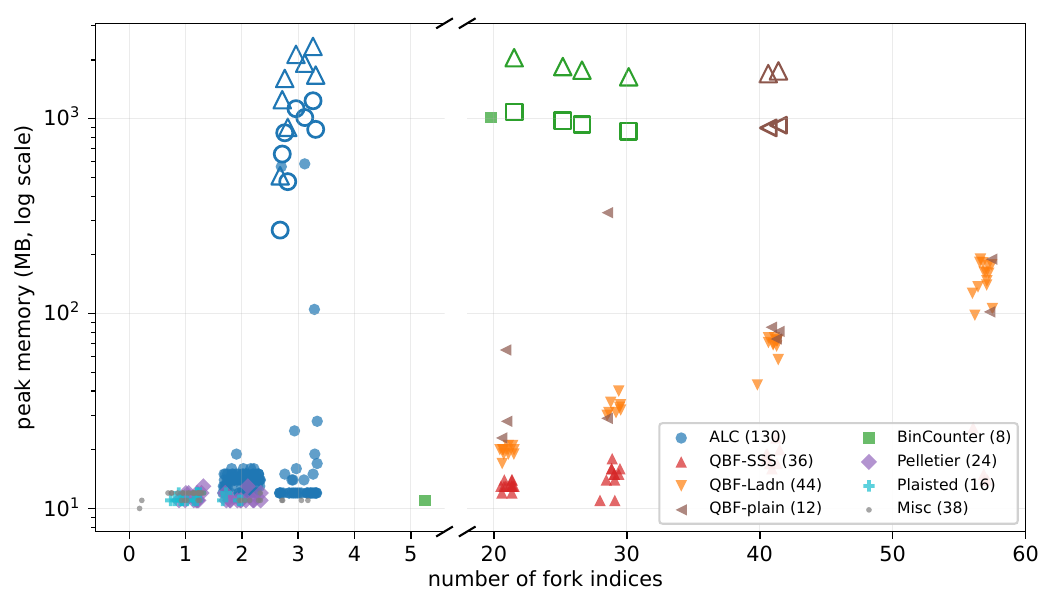}
  \caption{Peak memory against the number of fork indices, by problem family,
    for Vampire's default strategy under an instruction limit of
    $1.28\cdot10^{11}$. Hollow markers denote runs that reached the limit,
    whose memory is therefore only a lower bound. Note the broken horizontal
    axis.}
  \label{fig:memory}
\end{figure}

\fi

\end{document}